\documentclass[a4paper,12pt]{amsart}
\usepackage{amsmath,amssymb,amsthm}
\usepackage{amscd}
\usepackage{mathrsfs}
\usepackage[usenames,dvipsnames]{color}
\usepackage{graphicx}
\usepackage[all]{xy}
\usepackage{eucal}
\usepackage{extarrows} 
\usepackage{tabularx}
\usepackage{tikz} 
\usepackage{xparse} 
\usepackage{colordvi}
\usepackage{multicol}
\usepackage{enumitem}
\usepackage[normalem]{ulem}
\usepackage{epsfig}
\usepackage{rotating}
 \usepackage{ifpdf}
  \ifpdf
   \usepackage[colorlinks,final,backref=page,hyperindex]{hyperref}
  \else
   \usepackage[colorlinks,final,backref=page,hyperindex,hypertex]{hyperref}
  \fi

\usepackage{xspace}

\usetikzlibrary{arrows,matrix}
\usetikzlibrary{decorations.pathmorphing}

\begin{document}

\def\Xint#1{\mathchoice
{\XXint\displaystyle\textstyle{#1}}%
{\XXint\textstyle\scriptstyle{#1}}%
{\XXint\scriptstyle\scriptscriptstyle{#1}}%
{\XXint\scriptscriptstyle\scriptscriptstyle{#1}}%
\!\int}
\def\XXint#1#2#3{{\setbox0=\hbox{$#1{#2#3}{\int}$ }
\vcenter{\hbox{$#2#3$ }}\kern-.58\wd0}}
\def\ddashint{\Xint=}
\def\dashint{\Xint-}

\def\Xsum#1{\mathchoice
{\XXsum\displaystyle\textstyle{#1}}%
{\XXsum\textstyle\scriptstyle{#1}}%
{\XXsum\scriptstyle\scriptscriptstyle{#1}}%
{\XXsum\scriptscriptstyle\scriptscriptstyle{#1}}%
\!\sum}
\def\XXsum#1#2#3{{\setbox0=\hbox{$#1{#2#3}{\sum}$ }
\vcenter{\hbox{$#2#3$ }}\kern-.51\wd0}}
\def\ddashsum{\Xsum=}
\def\dashsum{\Xsum-}

\newcommand\cutoffint{\mathop{-\hskip -4mm\int}\limits}
\newcommand\cutoffsum{\mathop{-\hskip -4mm\sum}\limits}
\newcommand\cutoffzeta{-\hskip -1.7mm\zeta} 
\newcommand{\goth}[1]{\ensuremath{\mathfrak{#1}}}
\newcommand{\bbox}{\normalsize {}%
        \nolinebreak \hfill $\blacksquare$ \medbreak \par}
\newcommand{\simall}[2]{\underset{#1\rightarrow#2}{\sim}}

\newtheorem{theorem}{Theorem}[section]
\newtheorem{lem}[theorem]{Lemma}
\newtheorem{coro}[theorem]{Corollary}
\newtheorem{problem}[theorem]{Problem}
\newtheorem{conjecture}[theorem]{Conjecture}
\newtheorem{prop}[theorem]{Proposition}
\newtheorem{propdefn}[theorem]{Proposition-Definition}
\newtheorem{lemdefn}[theorem]{Lemma-Definition}
\theoremstyle{definition}
\newtheorem{defn}[theorem]{Definition}
\newtheorem{rk}[theorem]{Remark}
\newtheorem{ex}[theorem]{Example}
\newtheorem{coex}[theorem]{Counterexample}
\newtheorem{algorithm}[theorem]{Algorithm}
\newtheorem{convention}[theorem]{Convention}
\newtheorem{principle}[theorem]{Principle}

\renewcommand{\theenumi}{{\it\roman{enumi}}}
\renewcommand{\theenumii}{{\alph{enumii}}}

\newenvironment{thmenumerate}
{\leavevmode\begin{enumerate}[leftmargin=1.5em]}{\end{enumerate}}

\newcommand{\nc}{\newcommand}
\newcommand{\delete}[1]{}
\newcommand{\aside}[1]{\delete{#1}}

\nc{\mlabel}[1]{\label{#1}}  
\nc{\mcite}[1]{\cite{#1}}  
\nc{\mref}[1]{\ref{#1}}  
\nc{\mbibitem}[1]{\bibitem{#1}} 

\delete{
\nc{\mlabel}[1]{\label{#1}  
{\hfill \hspace{1cm}{\bf{{\ }\hfill(#1)}}}}
\nc{\mcite}[1]{\cite{#1}{{\bf{{\ }(#1)}}}}  
\nc{\mref}[1]{\ref{#1}{{\bf{{\ }(#1)}}}}  
\nc{\mbibitem}[1]{\bibitem[\bf #1]{#1}} 
}

\newcommand{\bottop}{\top\hspace{-0.8em}\bot}
\nc{\mtop}{\top\hspace{-1mm}}
\nc{\tforall}{\text{ for all }}
\nc{\bfcf}{{\calc}}
\nc{\calj}{{\mathcal J}}
\nc{\call}{{\mathcal L}}
\nc{\calr}{{\mathcal R}}

\newcommand{\R}{\mathbb{R}}
\newcommand{\FC}{\mathbb{C}}
\newcommand{\K}{\mathbb{K}}
\newcommand{\Z}{\mathbb{Z}}
\nc{\PP}{\mathbb{P}}
\newcommand{\Q}{\mathbb{Q}}
\newcommand{\N}{\mathbb{N}}
\newcommand{\F}{\mathbb{F}}
\newcommand{\T}{\mathbb{T}}
\newcommand{\G}{\mathbb{G}}
\newcommand{\C}{\mathbb{C}}

\newcommand{\bfc}{\calc}

\newcommand{\bbR}{\mathbb{R}}
\newcommand{\bbC}{\mathbb{C}}
\newcommand {\frakb}{{\mathfrak {b}}}
\newcommand {\frakc}{{\mathfrak {c}}}
\newcommand {\frakd}{{\mathfrak {d}}}
\newcommand {\fraku}{{\mathfrak {u}}}
\newcommand {\fraks}{{\mathfrak {s}}}
\newcommand {\frakS}{{\mathfrak {S}}}

\newcommand {\cala}{{\mathcal {A}}}
\newcommand {\calb}{\mathcal {B}}
\newcommand {\calc}{{\mathcal {C}}}
\newcommand {\cald}{{\mathcal {D}}}
\newcommand {\cale}{{\mathcal {E}}}
\newcommand {\calf}{{\mathcal {F}}}
\newcommand {\calg}{{\mathcal {G}}}
\newcommand {\calh}{\mathcal{H}}
\newcommand {\cali}{\mathcal{I}}
\newcommand {\callt}{{\mathcal {L}_\otimes}}
\newcommand {\calm}{{\mathcal {M}}}
\newcommand {\calp}{{\mathcal {P}}}
\newcommand {\cals}{{\mathcal {S}}}
\newcommand {\calv}{{\mathcal {V}}}

\newcommand{\conefamilyc}{{\mathfrak{C}}}
\newcommand{\conefamilyd}{{\mathfrak{D}}}

\newcommand{\Hol}{\text{Hol}}
\newcommand{\Mer}{\text{Mer}}
\newcommand{\lin}{\text{lin}}
\nc{\Id}{\mathrm{Id}}
\nc{\ot}{\otimes}
\nc{\bt}{\boxtimes}
\nc{\id}{\mathrm{Id}}
\nc{\Hom}{\mathrm{Hom}}
\nc{\im}{{\mathfrak Im}}
\nc{\supp}{\mathrm{Dep}}

\nc{\ola}[1]{\stackrel{#1}{\longrightarrow}}

\newcommand{\tddeux}[2]{\begin{picture}(12,5)(0,-1)
\put(3,0){\circle*{2}}
\put(3,0){\line(0,1){5}}
\put(3,5){\circle*{2}}
\put(3,-2){\tiny #1}
\put(3,4){\tiny #2}
\end{picture}}

\newcommand{\tdtroisun}[3]{\begin{picture}(20,12)(-5,-1)
\put(3,0){\circle*{2}}
\put(-0.65,0){$\vee$}
\put(6,7){\circle*{2}}
\put(0,7){\circle*{2}}
\put(5,-2){\tiny #1}
\put(6,5){\tiny #2}
\put(-5,8){\tiny #3}
\end{picture}}

%
%

\nc{\mge}{_{bu}\!\!\!\!{}}

\nc{\vep}{\varepsilon}

\def \e {{\epsilon}}
\nc{\syd}[1]{}
\newcommand{\cy}[1]{{\color{cyan}  #1}}
\newcommand{\sy}[1]{{\color{purple}  #1}}
\newcommand{\zb }[1]{{\color{blue}  #1}}
\newcommand{\zbd}[1]{}
\newcommand{\li}[1]{{\color{red} #1}}
\newcommand{\lit}[2]{\sout{\color{red}{#1}}{\color{red} #2}}
\newcommand{\lir}[1]{{\it\color{red} (Li: #1)}}

\nc{\prt}{P-}
\nc{\Prt}{P-}

\newcommand{\roc}{${\mathcal R}$\xspace}
\newcommand{\loc}{locality\xspace}
\newcommand{\Loc}{Locality\xspace}
\nc{\orth}{orthogonal\xspace}
\nc{\tloc}{L-}
\newcommand{\lset}{{\bf \loc SET }}
\newcommand{\set}{{\bf SET }}
\newcommand{\xat}{{X^{_\top 2}}}
\newcommand{\xbt}{{X^{_\top 3}}}
\newcommand{\xct}{{X^{_\top 4}}}

\newcommand{\gat}{{G^{_\top 2}}}
\newcommand{\gbt}{{G^{_\top 3}}}
\newcommand{\gct}{{G^{_\top 4}}}
\nc{\gnt}{{G^{_\top n}}}

\newcommand{\htwot}{{\calh ^{\ot_\top 2}}}
\newcommand{\hbt}{{calh ^{\ot_\top 3}}}
\newcommand{\hct}{{\calh ^{\ot_\top 4}}}

\title{Renormalisation via locality morphisms}

\author{Pierre Clavier}
\address{Institute of Mathematics,
University of Potsdam,
D-14476 Potsdam, Germany}
\email{clavier@math.uni-potsdam.de}

\author{Li Guo}
\address{Department of Mathematics and Computer Science,
         Rutgers University,
         Newark, NJ 07102, USA}
\email{liguo@rutgers.edu}

\author{Sylvie Paycha}
\address{Institute of Mathematics,
University of Potsdam,
D-14469 Potsdam, Germany\\ On leave from the Universit\'e Clermont-Auvergne\\
Clermont-Ferrand, France}
\email{paycha@math.uni-potsdam.de}

\author{Bin Zhang}
\address{School of Mathematics, Yangtze Center of Mathematics,
Sichuan University, Chengdu, 610064, China}
\email{zhangbin@scu.edu.cn}

\begin{abstract} This is a survey on renormalisation in the locality setup  highlighting the role that locality morphisms can play  for renormalisation purposes. Having set up a  general framework to build regularisation maps, we illustrate renormalisation by locality algebra homomorphisms on three examples, the renormalisation at poles of conical zeta functions, branched zeta functions  and iterated integrals stemming from Kreimer's toy model.
\end{abstract}

\subjclass[2010]{08A55,16T99,81T15, 32A20, 	52B20}

\keywords{locality, renormalisation, partial algebra, operated algebra, Hopf algebra, Rota-Baxter algebra, symbols}

\maketitle

\tableofcontents

\section*{Introduction}

The study of certain models, might they stem from quantum field theory or  be of a pure mathematical nature like that of measuring the volume of a convex cone,   gives rise to  {\em formal expressions}. {    Their evaluation  can yield infinities} instead of the desired numerical invariants, a situation which   calls for renormalisation. Getting rid of the infinities  first
calls for a {\em regularisation} which turns the formal expressions into meaningful regularised expressions, typically  described in terms of an algebra homomorphism
\begin{equation}\label{eq:phireg}
 \phi^{\rm reg}: A\longrightarrow M
\end{equation}
defined on a certain algebra $A$ (the algebra structure reflects the structure of the family of formal expressions) with values in a space $M$ of meromorphic functions. Next, these meromorphic functions are processed to only keep the holomorphic parts which can then be evaluated at zero, giving rise to the {\em renormalised values} of the formal expressions.

A basic requirement is that the renormalised values obey a multiplicative property, a requirement  reminiscent of the  {\em  \loc principle} in quantum field theory. It states that an object is only directly influenced by its immediate surroundings which in practice translates to the independence and multiplicativity of measurements of independent events.

In the case of an univariate regularisation,  $M=\C[\vep^{-1},\vep]]$ and projecting onto the holomorphic part $M_+=\C [[\vep]]$ (by means of the minimal subtraction scheme) gives rise to a Rota-Baxter operator $\pi_+: M\longrightarrow M_+$. Yet the Rota-Baxter operator itself is not
 {multiplicative} so that a mere minimal subtraction scheme $\pi_+\circ \phi^{\rm reg}$ does not preserve multiplicativity.  However, if the space $A$ carries  a suitable Hopf algebra structure, an algebraic Birkhoff factorisation \`a la  Connes and Kreimer~\cite{CK} implemented on the regularised  map $\phi^{\rm reg}:A\longrightarrow M$  guarantees the multiplicativity of the renormalised map  $\phi_+^{\rm ren}: A\longrightarrow M_+$, that is, the conservation of products after
renormalisation.

An alternative to univariate regularisation is  multivariate regularisation, a setup which enables to encode locality as a guiding principle, thus opening the way to new opportunities, and new challenges. One the one hand, the
multivariate minimal subtraction scheme, when available, is generally  not an algebra homomorphism, and   does not even give  rise to a Rota-Baxter operator, a major obstacle to the implementation of an algebraic Birkhoff
factorisation even when $A$ carries a Hopf algebra structure. Yet on the other hand, the very multivariate nature  of the scheme provides an easy book keeping device to preserve products for certain pairs of elements, we call
{\em independent pairs};  in accordance with the locality principle in quantum field theory, multiplicativity holds for independent pairs of arguments. The latter property, together with the multiplicativity of the evaluation map at zero, allows the renormalisation to preserve products for   independent pairs of elements.

An algebraic formulation of the locality principle in renormalisation  was discussed in \cite{CGPZ1}. There, we  express locality as a symmetric binary relation, study locality versions of algebraic structures, and develop a machinery used to preserve locality during the renormalisation procedure. It turns out that the locality setup is important not only in renormalisation, but also crucial in exploring deeper structures as we can see in the example of lattice cones.

The key asset of a locality setup lies in the fact that  the minimal subtraction scheme can be viewed as {\em locality algebra homomorphism}. For  example, the algebra $\calm$ of multivariate meromorphic germs with linear poles carries a locality algebra structure $(\calm , \perp ^Q, \cdot)$ and the corresponding projection $\pi_+$ from the decomposition
$\calm =\calm_+\oplus \calm _-^Q$ to $\calm _+$ (which as we stressed  previously is not an algebra homomorphism)   is a locality algebra homomorphism. Thus, for a regularisation, provided  the source space $A$ can be equipped with a  {\em locality algebra} $(A, \top _A)$ structure, and provided  the regularised map  $\phi^{\rm reg}$  takes its values in $\calm $ and it is a locality algebra homomorphism, then a  {\em  multivariate subtraction scheme}  can be implemented on the regularised maps $\phi^{\rm reg}: (A, \top _A) {\longrightarrow} (\calm , \perp ^Q)$. In this \loc setup, the renormalised map $\pi _+\circ \phi^{\rm reg}$ still do not preserve products yet they   preserve partial products, which is what one needs.

To sum up, we work in a setup  which encompasses the principle of locality; locality detects  pairs of independent elements of $A$ and partial multiplicativity amounts to multiplicativity on pairs
of  {independent} elements. The purpose of this survey based on previous work by the authors \cite{CGPZ1,CGPZ2, CGPZ3} and \cite{GPZ1,GPZ2} is i) to demonstrate how to  achieve a multivariate regularisation of a formal expression so as to build  a \loc algebra homomorphism (so a multivariate version of (\ref{eq:phireg}))
\begin{equation}\label{eqphimreg}\phi^{\rm reg}:  {(A, \top _A)}\longrightarrow \left(\calm, \perp^Q\right)\end{equation}
on a \loc algebra $(A, \top^A)$ with values in the \loc algebra $(\calm , \perp^Q) $, and ii) to renormalise the resulting regularised locality algebra homomorphism, describing the general theory and  illustrating it by examples. Let us describe  the contents  of the paper in more detail.

Since there are different approaches to explore  locality, in Section~\ref{sec:parloc}, much of which is borrowed from \cite{Zh}, we review and compare various partial structures with the locality structures introduced in
\cite{CGPZ1}. In particular,  we view the \loc setup    as a  symmetric version of the more general \roc-setup which comprises  partial semigroups (Definition \ref{defn:SemigroupSch}) introduced in \cite{Sch} and we relate \roc-monoids to the selective category of Li-Bland and Weinstein~\cite{LW} with one object (Proposition~\ref{lem:selcatLWA}).  Thereafter, for the sake of simplicity, we choose to keep to the \loc setup which turns out to be sufficient for the renormalisation purposes we have in mind.

In Section~\ref{sec:localg}, we introduce \loc algebras  (Definition \ref{defn:localgebra}), a notion we first illustrate  by the pivotal  example of $\R^\infty$ (Example \ref{ex:Rinfty}) equipped with an inner product $Q$ which
induces an orthogonality relation $\perp^Q$, after which we discuss in Paragraph \ref{subsec:meroloc},   the algebra  $\calm$ of multivariate meromorphic germs with linear poles at zero,  equipped with a \loc relation induced by $\perp^Q$, which by a slight abuse of notation is denoted by the same symbol  (see Proposition \ref{prop:MCperp}).  Other relevant examples are the \loc algebra of  {\em  lattice cones } in Paragraph \ref{subsec:latticecones} and the \loc algebra of {\em  properly decorated rooted forests} in Paragraph \ref{subsec:forestsloc}.

Section~\ref{sec:locmor} is dedicated to  the main protagonists of this paper, namely   {\em  \loc morphisms} (Definition \ref{defn:locallmap}) of \loc algebras, so maps between \loc algebras which, as well as preserving   the \loc relation and \loc vector space structure, further preserve  the related partial product.

Amongst these are  {\em \loc projections } $\pi_+^Q:\calm\longrightarrow \calm_+$ onto the space $\calm_+$ of holomorphic germs at zero  built from the inner product $Q$. Such projections  arise from the decomposition
$\calm=\calm_+ \oplus \calm_-^Q$ (Eq. ~(\ref{eq:decmero})) induced by $Q$, and their \loc as morphisms is a consequence of  the fact that $\calm_+$ (resp. $\calm_-^Q$) is a \loc subalgebra (resp. \loc ideal) of $\calm$,
(Proposition \ref{prop:m-}).

Composed with the evaluation at zero ${\rm ev}_0$ these projections yield   useful  {\em  renormalisation schemes} discussed in Paragraph \ref{subsec:renscheme}:  \begin{equation}\label{eq:evpi}{\rm ev}_0\circ \pi_+^Q:  \calm  \longrightarrow \C,\end{equation} which can be viewed as a  {\em  multivariate minimal subtraction scheme}.

With this multivariate minimal subtraction scheme, a renormalisation process is reduced to two steps:
\begin{enumerate}
         \item to construct the regularised map $\phi^{\rm reg}:  {(A, \top _A)}\longrightarrow \left(\calm, \perp^Q\right)$;
         \item to  implement  renormalisation schemes of the type (\ref{eq:evpi}) to  the regularised map \hfill \break \noindent  $\phi^{\rm reg}:  {(A, \top _A)}\longrightarrow \left(\calm, \perp^Q\right)$ in order to build  the {\em renormalised map} $\phi^{\rm ren} := {\rm ev}_0\circ \pi_+^Q \circ \phi^{\rm reg}: A\to \C .$
       \end{enumerate}

Various \loc maps built in Section \ref{sec:locmor} are interpreted in Section \ref{sec:renorm} as regularisation maps $\phi^{\rm reg}: A\longrightarrow \calm$ which need to be renormalised,
all of which stem from formal sums and integrals as multivariate regularisations.

We first illustrate (in Paragraphs~\ref{subs:LocCones} and \ref {ss:renczv}) this multivariate approach  with {\em conical zeta functions} (resp. {\em branched zeta functions}), which to a lattice cone (resp. a decorated rooted forest), assign a renormalised value of the
{\em  regularised conical zeta function} (resp.  {\em regularised   branched zeta function}) at poles. The (partial) multiplicativity of the maps encoded in their very construction in our multivariate \loc setup, ensures their multiplicativity on orthogonal lattice cones (resp. independent decorated rooted forests).

In \cite{CGPZ1,GPZ1}, {\em conical zeta functions} (Paragraph~\ref {ss:renczv}), which generalise multiple zeta functions were built using exponential sums on lattice cones.
The exponential sum $S$, resp.  integral $I$  on a lattice cone correspond to the discrete, resp. continuous Laplace transformation of the characteristic function of the lattice cone (Proposition \ref{prop:SI}). One easily checks that   Laplace transforms of characteristic functions of smooth cones define meromorphic maps with linear poles;  the fact that $S$ and $I$ take their values in $\calm$ for any convex lattice cone, then follows from their additivity on disjoint unions combined with the fact that any convex lattice cone can be subdivided in smooth lattice cones. Both maps  define \loc algebra homomorphisms on the \loc algebra of lattice cones for a \loc relation  induced by the orthogonality relation $\perp^Q$ on $\R^\infty$. Their multiplicativity on orthogonal lattice cones follows from the usual homomorphism property of the exponential map on these cones.

A second example  which provides an  alternative  generalisation of multiple zeta functions, is given by {\em branched zeta functions}~\cite{CGPZ2} (discussed in Paragraph \ref{subsec:branchedzeta}) associated with rooted forests (Paragraph \ref{subsec:forestsloc}). These are built by means of a
 {\em  branching procedure} which strongly relies on  the universal properties of  properly decorated rooted forests (see Proposition~\ref{pp:universal}).  Such a branching procedure lifts  a map  $\phi$ defined on the decoration set to what we call a   {\em  branched map}  $\widehat \phi $ on the algebra of decorated forests
(see (\ref{eq:hatphi})). Applied to  a   {summation map} $\phi=\frakS _\lambda$ on the locality algebra $\Omega$ of meromorphic germs of symbols, this branching procedure gives rise to  a branched sum ${\widehat{\frakS _\lambda}}$ acting on the algebra of properly decorated rooted forests by meromorphic family of symbols on $\R_{\geq 0}$. The universal property  underlying the construction ensures the  multiplicativity on independent   {forests}. Combining this with the \loc morphism given by the Hadamard finite part at infinity ({\ref{eq:finitepart_const})-a linear form on polyhomogeneous pseudodifferential symbols which coincides on smoothing symbols with the limit at infinity-  extended to $\Omega$, gives rise to   {\em  branched regularised zeta functions}  $\zeta^{\rm reg, \lambda}$ defined on the \loc algebra of properly $\Omega $-decorated rooted forests.

In Paragraph \ref {subs:operation} we describe similar constructions based on the universal properties of properly decorated rooted forests~\cite{CGPZ3}, which yield a third example (Paragraph~\ref{subsec:Kreimer}), namely $\calm$-valued maps stemming from iterated integrals  arising in Kreimer's toy model~\cite{K}.

To sum up, the \loc setup combined with the multivariate regularisation provides a way to preserve (partial) multiplicativity while renormalising, in accordance with the \loc principle  in physics.

\section{Partial  versus locality structures}
\label{sec:parloc}
We review and compare various partial product structures with the locality structures introduced in \cite{CGPZ1}; although the concept of  algebraic \loc structures   is to our knowledge new in the context of renormalisation,  partial products have been used in other contexts, hence the need to relate the two concepts, partial and \loc products. This section is based on \cite{Zh}.

\subsection{Partial semigroups}

We start with a generalisation of the notion of a \loc set introduced in \cite{CGPZ1}, by dropping the symmetry  property of the relation required in \cite{CGPZ1}:
\begin{defn}
 \begin{enumerate}
  \item An {\bf \roc-set } is a couple $(X,\top)$ with $X$ a set and $\top\subset X\times X$ a binary relation on $X$. We also write $X\times_\top X$ for $\top$.
  \item Let $(X,\top)$ be an \roc-set and $U\subset X$. We write ${}^\top U$ (resp. $U^\top$) the {\bf left polar set} (resp. {\bf right polar set})
  of $U$; defined by
  \begin{equation}\label{eq:Ul}
   {}^\top U : =\{x\in X|\,(x,u)\in\top~\forall u\in U\}
  \end{equation}
  (resp.
  \begin{equation}\label{eq:Ur}
   U^\top : =\{x\in X|\,(u,x)\in\top~\forall u\in U\}).
  \end{equation}
 \end{enumerate}
If $\top$ is a symmetric binary relation, we call, as in~\cite{CGPZ1}, the couple $(X, \top)$ a {\bf \loc set}, in which case  $ {}^\top U= U^\top$.
\end{defn}
Let {\bf RS} (resp. {\bf LS}) denote the category of \roc-sets (resp. \loc  sets) whose morphisms are maps $\phi: (X,\top_X)\longrightarrow (Y, \top_Y)$ such that $(\phi\times \phi)(\top_X)\subset \top_Y$, called {\bf \roc-maps} (resp. {\bf \loc maps}).

We equip an {\bf \roc-set} with four distinct, however related, partial product structures, the first one is a generalisation (dropping the symmetry condition) taken from \cite{Zh} of the locality relation introduced in \cite{CGPZ1}:

\begin{defn}
 An {\bf \roc-semigroup} is an \roc-set $(X,\top)$ together with a partial product map
 \begin{align*}
  \mu:\top & \mapsto X \\
     (x,y) & \mapsto x\,y
 \end{align*}
 which we denote by $(X, \top, \mu)$, such that:
 \begin{enumerate}
  \item For any subset $U\subset X$, \begin{equation}\label{eq:muUl}\mu(({}^\top U\times {}^\top U)\cap\top)\subseteq{}^\top U, \end{equation}
  \item For any subset $U\subset X$, \begin{equation}\label{eq:muUr}\mu((U^\top\times U^\top)\cap\top)\subseteq U^\top. \end{equation}
  \item For any $a,b,c$ in $X$ such that any couple lies in $\top$ we have $(a\,b)\,c=a\,(b\,c)$.
 \end{enumerate}
If $\top$ is a symmetric binary relation,  condition (\ref{eq:muUl}) coincides with (\ref{eq:muUr}) and we call $(X, \top, \mu)$ a {\bf \loc semigroup}.

Let us denote by {\bf RSg} (resp. {\bf LSg}) the category of \roc- (resp. \loc) semigroups whose morphisms are \roc-maps (resp. \loc maps)
\[\phi: (X,\top_X, \mu_X)\longrightarrow (Y, \top_Y, \mu_Y),\]
which are partially multiplicative
\[(a, b)\in \top_X\Longrightarrow \phi( \mu_X(a, b))=  \mu_Y\left(\phi(a), \phi(b)\right).\]
They are called {\bf  \roc-morphisms} (resp. {\bf \loc morphisms}).
\end{defn}

\begin{rk} Note that a map between  two \loc semigroups is a \loc morphism  if and only if it is an \roc-morphism.
\end{rk}

\begin{rk}
It is easy to check that
\begin{itemize}
\item Eq.~(\ref{eq:muUl}) is equivalent to   \begin{equation}\label{eq:muxl}\left(x\top z\wedge y\top z\wedge x\top y\right)\Longrightarrow (x\, y)\top z \quad\forall (x, y, z)\in X^3, \end{equation}
\item Eq.~(\ref{eq:muUr}) is equivalent to   \begin{equation}\label{eq:muxr}\left(z\top x\wedge z\top y\wedge x\top y\right)\Longrightarrow z\top (x\, y)\quad \forall (x, y, z)\in X^3. \end{equation}
\end{itemize}
\end{rk}
The following definitions are taken from \cite{Zh}.
\begin{defn} (see \cite[Definition 3.1]{Zh})
\begin{enumerate}
\item A {\bf strong \roc-semigroup} is an \roc-set $(X,\top)$ together with a partial product map
 \begin{align*}
  \mu:\top & \mapsto X \\
     (x,y) & \mapsto x\,y
 \end{align*}
  also denoted by $(X, \top, \mu)$, such that for any $x, y, z\in X$ :
 \begin{equation*}
  \left((x,y)\in\top\wedge (y,z)\in\top \right) \Longrightarrow \left((x\,y,z)\in\top\wedge (x,y\,z)\in\top \wedge (x\, y)\, z= x\, (y\, z)\right).
 \end{equation*}
 Let us denote by  {\bf SRSg}   the category of strong \roc-semigroups whose morphisms are   {\bf \roc-morphisms}.
 \item
 A {\bf refined \roc-semigroup} is an  \roc-set $(X,\top)$ together with a partial product map
 \begin{align*}
  \mu:\top & \mapsto X \\
     (x,y) & \mapsto x\,y
 \end{align*}
 such that:
 \begin{itemize}
  \item $(x,y)\in\top\Longrightarrow\left((y,z)\in\top\Leftrightarrow(x\, y,z)\in\top,~\forall z\in X\right)$,
  \item $(y,z)\in\top\Longrightarrow\left((x,y)\in\top\Leftrightarrow(x,y\,z)\in\top,~\forall x\in X\right)$,
  \item For any $(x,y)\in\top$ and $(y,z)\in\top$ we have $(x\, y)\, z=x\,(y\,z)$.
 \end{itemize}
 Let us denote by {\bf RRSg} the category of refined \roc-semigroups whose morphisms are \roc-morphisms.
 \end{enumerate}
\end{defn}
\begin{rk} (\cite[Proposition 3.3]{Zh}) Every strong \roc-semigroup is clearly an \roc-semigroup, but the converse does not hold. See e.g. \cite[Counterexample 3.4]{Zh} and the subsequent paragraph.
\end{rk}
The following definition is taken from \cite{Sch}. See also \cite[Definition 2.20]{Zh}.
\begin{defn}\label{defn:SemigroupSch}
 A {\bf partial semigroup} is an \roc-set $(X,\top)$ together with a partial product map
 \begin{align*}
  \mu:\top & \mapsto X \\
     (x,y) & \mapsto x\,y
 \end{align*}
 such that for any $x,y,z\in X$
 \begin{equation}\label{eq:Passoc}
  \left((x,y)\in\top\wedge (x\, y,z)\in\top \right)\Leftrightarrow ((y,z)\in\top\wedge (x,y\,z) \in\top)
 \end{equation}
 in which case $(x\,y)\, z=x\,(y\, z)$ also holds. Let us denote by {\bf PSg} the category of partial semigroups whose morphisms are   {\bf \roc-morphisms}.
\end{defn}

The notion of partial semigroup relates to a  particular instance of the 
{\bf selective category} of Li-Bland and Weinstein introduced in \cite[Definition 2.1]{LW}, whose definition we now recall.

\begin{defn}  A {\bf selective category }
is a category ${\mathcal C}$ whose set of morphisms (resp. objects)  we denote by ${\rm Mor} $  (resp. ${\rm Ob} $) together with a distinguished class ${\mathcal S}\subset{\rm Mor}$ of morphisms, called
{\bf suave}, and a class $\top_ {\mathcal S}\subset {\mathcal S} \times {\mathcal S} $ of pairs of suave morphisms called {\bf congenial pairs}, such that:
\begin{enumerate}
\item Any identity morphism is suave so ${\rm Id}_x  $  is suave for any $x\in {\rm Ob} $ which we write for short ${\rm Id}\subset {\mathcal S}$;
\label{it:selcat1}
\item If {$f:X\longrightarrow Y$ is suave, $({\rm Id}_Y, f)$ and $(f, {\rm Id}_X)$} are congenial;
\label{it:selcat2}
\item  If $f$ is a suave isomorphism, its inverse $f^{-1}$ is suave as well, and the pairs $(f,f ^{-1})$  and
$(f^{-1},f) $ are both congenial;
\label{it:selcat3}
\item  If $f$ and $g$ are suave and
$(f,g)$  is congenial, then $f\circ g$ is suave, i.e., the composition is a map $\circ:  \top_ {\mathcal S}\longrightarrow {\mathcal S} $;
\label{it:selcat4}
\item  If $f,g,h\in \mathcal{S}$,  then \[\left((f,g) \in \top_{\mathcal S}\wedge (f\circ g,h)\in \top_ {\mathcal S}\right) \Longleftrightarrow \left((g,h) \in \top_ {\mathcal S}\wedge (f, g\circ h)\in \top_ {\mathcal S}\right) ,\] in which case    $(f,g,h)$ is called a {\bf congenial triple}.
\label{it:selcat5}
\end{enumerate}
A {\bf selective  functor} between selective categories is one which takes congenial pairs to congenial pairs.
\end{defn}

Recall that a category $\mathcal{C}=({\rm Obj}(\calc),{\rm Mor}(\calc))$ is {\bf small} if ${\rm Obj}(\calc)$ and ${\rm Mor}(\calc))$ are sets and not proper classes.

\begin{prop}\label{lem:selcatLWA} A small selective category with one object reduces to a   partial  semigroup $(S, \top\subset S\times S, m)$  built from a nonempty subset $S\subset M$ of a  monoid $(M, \mu)$ with unit $1$ such that,
\begin{enumerate}
\item  $1\in S$, $1\top S$ and $S\top 1$;
\item  $(S,\top)$ is stable under taking inverse (in $M$) in the following sense: if $s\in S$  is invertible in $M$, then its inverse $s^{-1}$ is in $S$ and $(s, s^{-1}), (s^{-1},s)$ are in $\top$.
\end{enumerate}
A selective morphism
between selective categories with one object reduces to  \roc- morphisms of partial semigroups that preserve the identity (and hence inverses).
\end{prop}
\begin{proof}
With exactly one object, a small category $\mathcal{C}=({\rm Obj}(\calc),{\rm Mor}(\calc))$ boils down to a monoid $M:={\rm  Mor}(\calc)$,  its distinguished class ${\mathcal S}{\rm Mor}$ of suave morphisms boils down to a subset $S\subset M$, and the class of congenial pairs of suave elements boils down to a subset $\top\subset S\times S$. Further conditions~(\ref{it:selcat1}) -- (\ref{it:selcat3}) of a selective category boil down to the two conditions in the lemma, while conditions~(\ref{it:selcat4}) -- (\ref{it:selcat5}) boil down to the condition that $(S,\top)$ is a partial semigroup.

Finally a selective  functor $f: (S_1, \top_1)\to (S_2, \top_2)$
between selective categories $(S_i, \top_i)$  with one object boils down to a \roc-morphism of partial semigroups that preserve the identity.
\end{proof}

\begin{rk}
     We need the category $\calc$ to be small, as even a category with only one object can be large. For example, take $\calc$ the category whose only object {\bf Set} is the category of sets, and whose morphisms are the endofunctors
     of {\bf Set}. In this example, Mor$(\calc)$ has no  monoid structure as it is not a set.
    \end{rk}

\subsection{Relating  various partial structures}

We quote from \cite{Zh} with the reference to the statements. We start with some general comparisons:
\begin{itemize}
 \item {\bf RRSg} $\subsetneq$ {\bf SRSg} \cite[Example 4.3]{Zh}.
 \item {\bf SRSg} $\subsetneq$ {\bf RSg} \cite[Proposition 3.3 and Counterexample 3.4]{Zh}.
 \item {\bf SRSg} $\subsetneq$ {\bf PSg} \cite[Example 3.6]{Zh}.
 \item {\bf SRSg} $\subsetneq$ {\bf RSg} $\cap$ {\bf PSg} \cite[Proposition 3.7]{Zh}.
\end{itemize}
There are examples of ${\mathcal R}$-semigroups  that are not partial semigroups and vice-versa:
\begin{itemize}
 \item {\bf RSg} $\not\subseteq$ {\bf PSg} \cite[Example 3.8]{Zh}.
 \item {\bf PSg} $\not\subseteq$ {\bf RSg} \cite[Example 3.10]{Zh}.
\end{itemize}

Note that the last two conditions mean that {\bf RSg}$\cap${\bf PSg} $\subsetneq${\bf PSg} and {\bf RSg}$\cap${\bf PSg}$\subsetneq${\bf RSg}. Thus in summary, we have strict inclusions shown by the following Hasse diagram.
$$
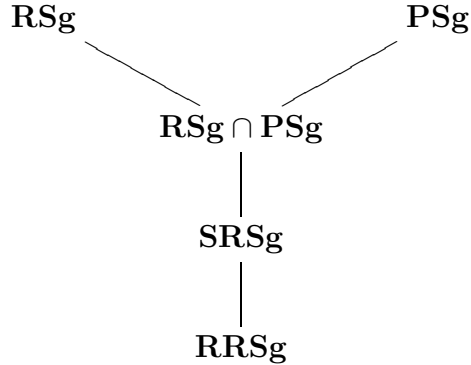
\xymatrix{
\mathbf{RSg} \ar@{-}[rd] && \mathbf{PSg} \ar@{-}[ld]\\
& \mathbf{RSg}\cap \mathbf{PSg} \ar@{-}[d] & \\
& \mathbf{SRSg} \ar@{-}[d] & \\
& \mathbf{RRSg} & }
$$

Here are  examples of \loc sets with a partial product   which   fulfills  the following equivalence relation:
\[(x\top y ~\wedge~ (x\, y)\top z)\Longleftrightarrow (x\top y ~\wedge~y\top z~\wedge~ x\top z) \Longleftrightarrow  (y\top z  ~\wedge~ x \top (y\, z)),\] namely,
 conditions  (\ref{eq:muxl}),   (\ref{eq:muxr})  {(which are equivalent for \loc semigroups)} are equivalent to (\ref{eq:Passoc}). So they are  both \loc and partial  semigroups.

\begin{ex}
\begin{enumerate}
  \item The set $\N$ of natural numbers equipped with the  coprime relation $n\top m\Leftrightarrow n\wedge m=1$  and the usual product of real numbers is  a partial semigroup since
  \[a\wedge b=1 \, {\rm and} \, a\,b\wedge c=1\Longleftrightarrow a\wedge b=1\, {\rm and} \,a\wedge c=1\, {\rm and}\, b\wedge c=1\Longleftrightarrow c\wedge b=1  \, {\rm and} \, a\wedge b\,c=1,\] and a \loc semigroup since   \[a\wedge c=1 \, {\rm and} \, b\wedge c=1\Longrightarrow a\,b\wedge c=1.\]
     \item The power set ${\mathcal P}(X)$ of a set $X$ equipped with the disjointness relation $A\top B\Leftrightarrow A\cap B=\emptyset$ and the product law given by the union $\cup$ is a partial semigroup and we have
    \begin{eqnarray*}
A\cap B=\emptyset \, \wedge \, (A\cup B)\cap C=\emptyset &\Longleftrightarrow& A\cap B=\emptyset \, \wedge \,  A\cap C=\emptyset  \, \wedge \, B\cap C=\emptyset\\
&\Longleftrightarrow &  B\cap C=\emptyset \, \wedge \,A\cap (B\cup C) =\emptyset.\end{eqnarray*} It is also a \loc semigroup since \[A\cap C=\emptyset \wedge  B\cap C=\emptyset \Longrightarrow (A\cup B)\cap C=\emptyset.\]
 \end{enumerate}
\end{ex}

\subsection{Transitive partial structures}

Here is a  useful property of partial structures.

\begin{defn}
 A \loc set $(X,\top)$ is called {\bf transitive} if the relation $\top$ is transitive, namely if for any $a,b,c\in X$
 \begin{equation*}
 \left( (a,b)\in\top\wedge(b,c)\in\top\right)\Longrightarrow (a,c)\in\top.
 \end{equation*}
 A partial structure $(X,\top,\mu)$ such that $(X,\top)$ is transitive is called a {\bf transitive partial structure}. We write {\bf tLSg} (resp.
 {\bf tSLSg, tRSg, tPSg}) for the category of transitive \loc semigroups (resp. transitive strong \loc semigroups, transitive refined \loc semigroups,
 transitive partial semigroups).
\end{defn}
\begin{rk} Transitive partial structures are not relevant in the context of locality understood in the sense of quantum field theory, since we do not expect the event $A$ to be independent of the event $C$ under the assumption
that  the  event $A$ is independent of the  event $B$ and  the event $B$ is independent of the event $C$. In fact, a transitive \loc structure $\top$ is {\em almost} reflexive, in that for every event $a$, if there exists $b$ such that $b\top a$, then $a$ is independent of itself.
\end{rk}

We saw that \loc semigroups and partial semigroups are distinct structures. However, we have the following result:
\begin{prop}\cite[Proposition 3.9]{Zh}
 {\bf tLSg} $\subsetneq$ {\bf tPSg}.
\end{prop}
The statement of \cite{Zh} involves a non-strict inclusion $\subseteq$, yet  \cite[Example 3.10]{Zh} gives a transitive partial semigroup which is not a \loc semigroup.

\section{\Loc algebras}
\label{sec:localg}
Throughout the paper we choose to work in the framework of \loc structures, partially for simplicity but more so due to the fact the the applications we have in view do not require the more general framework of \roc-structures.

\subsection{Basic definitions} We borrow the subsequent definitions from \cite{CGPZ1}. Among them, locality algebras are fundamental objects in multivariate renormalisation.

\begin{defn}\begin{enumerate}\item
A {\bf \loc vector space} is a vector space $V$ over a field $K$ equipped with a \loc relation $\top$ which is compatible with the linear structure on $V$ in the sense that, for any  subset $X$ of $V$, $X^\top$ is a linear subspace of $V$.
\item A {\bf \loc   monoid} is a \loc   semigroup $(G,\top, m_G)$ together with a {\bf unit element} $1_G\in G$ given by the defining property
\[\{1_G\}^\top=G\quad \text{ and }\quad m_G(x, 1_G)= m_G(1_G,x)=x\quad \tforall  x\in G.\]
\item     A {\bf (resp. unital) \loc  algebra} $(A,\top,+,\cdot, m_A)$ (resp. $(A,\top,+,\cdot, m_A, 1_A)$) over $K$ is a \loc vector space $(A,+,\cdot,\top)$ over $K$ together with a \loc bilinear map
	$$ m_A: A\times_\top A \to A$$ such that
	$(A,\top, m_A)$ is a \loc semigroup (resp. a \loc monoid  with   {\bf unit} $1_A\in A$).
$(A,\top)$ is called {\bf commutative} if $(A,\top,m_A)$ is a commutative \loc semigroup.
\item A {\bf sub-\loc algebra} of a \loc algebra $(A,\top,m_A)$ is a linear subspace $B$ of $A$ such that with respect to the \loc condition $\top_B:=(B\times B)\cap \top$ of $\top$ and the partial product $m_B:=m_A|_{\top_B}$ on $B$, $(B,\top_B,m_B)$ is a \loc algebra.
\item
A sub-\loc algebra $I$ of a \loc commutative algebra $\left(A,\top ,m_A \right) $ is called a {\bf \loc ideal} of $A$ if  {$m_A(I^\top,I)\subseteq I$; i.e. if} for any $b\in I$ we have
$m_A(c, b)\in I$ for all $c\in\{b\}^\top$.
\end{enumerate}
\label{defn:localgebra}
\end{defn}

\begin{ex}\label{ex:Rinfty}A pivotal example is the \loc vector space $\left(\R ^\infty, \perp^Q\right)$,  where $\R ^\infty=\bigcup_{k\geq 1} \R ^k$ and $Q=(Q_k(\cdot, \cdot))_{k\geq 1}$  is the inner product on $\R^\infty$  defined by  the inner products on $\R^k$
$$ Q_k(\cdot,\cdot): \R ^k\otimes \R ^k \to \R, \quad k\geq 1,$$ such that $Q_{k+1}|_{\R^k\ot \R ^k}=Q_k$. The inner product induces a \loc relation on $\R ^\infty$
$$u\perp ^Q v \Leftrightarrow Q(u,v)=0,
$$
which makes $\left(\R ^\infty, \perp^Q\right)$ a \loc vector space.

The inner product also induces a locality set structure on the set of subspaces of $\R ^\infty$:
$$U\perp ^Q V \Leftrightarrow Q(u,v)=0, \forall u\in U, v\in V.
$$
\end{ex}

\subsection {The locality algebra of meromorphic germs with linear poles}\label{subsec:meroloc}
Recall that for the filtered Euclidean space $\left(\R ^\infty, Q\right)$
from the standard embeddings $i_n: \R ^n \to \R ^{n+1}$,
the inner product $Q$ induces an isomorphism
$$Q_n^*: \R ^n \to (\R ^n)^*.
$$
So
$$j_{n+1}:=(Q_n^*)^{-1} \circ i_n^*\circ Q_{n+1}^*: \R ^{n+1} \to \R ^n,
$$
induce a direct system
$$j_{n+1}^*: \calm (\R ^n \otimes \C)\to \calm (\R ^{n+1}\otimes \C),
$$
and we set
\begin{equation}
\calm :=\calm (\bbC ^\infty):=\varinjlim _n \calm (\bbC ^n)=\varinjlim _n \calm (\bbR ^n\otimes \bbC)
\label{eq:calm}
\end{equation}
to be the algebra of multivariate meromorphic  germs  with linear poles and real coefficients ~\mcite{GPZ1,GPZ2}.

The \loc structure on $\left(\R ^\infty, \perp^Q\right)$  induces a locality structure on $\calm $. For $f\in \calm (\bbC ^n) $, let $\supp(f)$ denote the {\bf dependence space} of $f$, defined as the smallest subspace of $(\C^n)^*$ spanned by the linear forms on which $f$  depends in the sense of \cite[Definitions 2.9 and 2.13]{CGPZ1}.

\begin{prop}\label{prop:MCperp} \cite[{Proposition 3.9}]{CGPZ1} Equipped with the \loc  relation
\[f_1\perp^Q f_2\Longleftrightarrow {\rm Dep}(f_1)\perp^Q{\rm Dep}(f_2),\]
and the ordinary product of functions restricted on the graph of the locality relation, the \loc set $\left(\calm, \perp^Q\right)$ carries a \loc algebra structure.
\end{prop}

The inner product $Q$ induces a decomposition of $\calm$ \cite {GPZ2}
\begin{equation}
\calm =\calm _+\oplus \calm ^Q_-,
\label{eq:decmero}
\end{equation}
where $\calm _+$ is the subspace of holomorphic germs and $\calm ^Q_-$ is the subspace generated by polar germs.

\begin {prop} {\cite[Proposition 3.19]{CGPZ1}} The subspace $\calm _+$ is a subalgebra and sub-locality algebra of $\calm$. The subspace $\calm ^Q _-$ is not a subalgebra but a locality ideal of $\calm$.
\mlabel {prop:m-}
\end{prop}

There is another locality structure on $\calm $ which is also compatible with the ordinary product of functions.

Let $\{e_n\,|\, n\in \N\}$ denote a $Q$-orthonormal basis of $\R^\infty$. We call the {\bf  support of $f\in \calm$}, denoted ${\rm Supp}(f) $, the smallest subset $J\subset \N$ such that ${\rm Dep}(f)$ is contained in the subspace spanned by $\{e_j^*\,|\, j\in J\}$. We thus equip  $\calm$ with  the \loc relation
\[f_1\top_D^Q f_2\Longleftrightarrow {\rm Supp}(f_1)\cap {\rm Supp}(f_2)=\emptyset,\]
which makes $\calm $ a \loc vector space.

\begin{rk} \label{rk:differences_loc_struct_mero} Since the $K$-linear span of ${\rm Supp}(f)$ contains ${\rm Dep}(f)$, for $f_1, f_2\in \calm$ we have $f_1\top^Q _D f_2\Longrightarrow f_1\perp^{Q} f_2$. Yet $(e_1^*+e_2^*)\perp^{Q} (e_1^*-e_2^*)$  whereas  these two linear forms are not $\top _D$ independent since  ${\rm Supp}(e_1^*+e_2^*)=\{1,2\}={\rm Supp}(e_1^*-e_2^*)$.
\end{rk}
\begin{prop}
The \loc set $\big(\calm, \top _D^Q\big)$ equipped with the product of functions is a \loc  algebra.
\end{prop}
\begin{proof}
     This follows from Remark \ref{rk:differences_loc_struct_mero} and Proposition \ref{prop:MCperp}.
    \end{proof}

\subsection {\Loc algebra of lattice cones}\label{subsec:latticecones}  In the filtered Euclidean lattice space $\left(\R ^\infty, \Z ^\infty, Q\right)$ such that $Q(u,v)\in \Q$ for $u,v \in \Z ^\infty$, a lattice cone is a pair $(C, \Lambda _C)$ where $C$ is a polyhedral cone in some $\R ^k$ generated by elements in $\Z ^\infty$ and $\Lambda _C$ is a lattice generated by elements in $\Q ^\infty $ in the linear subspace spanned by $C$.
Let $\bfcf _k$ be the set of lattice cones in $\R ^k$ and
 $$\bfcf =\bigcup_{k\geq 1} \bfcf_k  $$
be the set of lattice cones in $(\R ^\infty, \Z ^\infty)$ which is the direct limit under the standard embeddings. Let $\Q \bfc _k $ and $\Q \bfc $ be the linear spans of $\bfcf _k $ and $\bfcf$ over $\Q$.

For two convex lattice cones $(C_i, \Lambda_i)$ we set
\begin{equation}\label{eq:perpcones}(C_1, \Lambda _1)\perp^Q (C_2, \Lambda _2) \Longleftrightarrow {\rm span} (C_1) \perp^Q {\rm span }(C_2).\end{equation}
This defines a locality relation on $\Q \bfc$.

For convex cones $C:=\langle u_1,\cdots,u_m\rangle$ and $D:=\langle v_1,\cdots,v_n\rangle$ spanned by $u_1,\cdots,u_m$ and $v_1,\cdots,v_n$ respectively, their Minkowski sum is the convex cone
\begin{equation*}
C\cdot D:=\langle u_1,\cdots,u_m,v_1,\cdots, v_n\rangle.
\end{equation*}
This operation can be extended to a product in $\Q \bfc$:
\begin{equation}
(C, \Lambda _C)\cdot (D, \Lambda _D):=(C\cdot D, \Lambda _C+\Lambda _D),
\mlabel{eq:minkprod}
\end{equation}
where $\Lambda _C+\Lambda _D$ is the abelian group generated by $\Lambda _C$ and $\Lambda _D$ in $\Q ^\infty$.
This product endows a monoid structure on $\bfc$ with unit $(\{0\}, \{0\})$, which also restricts to a \loc monoid structure on $(\bfcf,\perp ^Q)$.

\begin {prop}  \cite[Lemma 3.18]{CGPZ1}
The \loc set $\left(\Q \bfc, \perp^Q\right) $  equipped with the Minkowski sum
is a (graded) locality algebra.
\end{prop}

As with the case for meromorphic germs, there is another subset $\top _D$ of $\perp ^Q$ which also makes $\Q\bfc$ into a locality algebra. Let $\{e_n\,|\, n\in \N\}$ be an orthonormal basis of $\R^\infty$. For a lattice cone $(C, \Lambda _C)$, we denote by ${\rm Supp}(C, \Lambda _C) $ the smallest subset $J$ such that ${\rm span}(C)$ is contained in the subspace spanned by $\{e_j^*\,|\, j\in J\}$ and equip  $\Q \bfc$ with  the \loc relation
\[(C_1, \Lambda _{C_1})\top_D^Q (C_2, \Lambda _{C_2})\Longleftrightarrow {\rm Supp}(C_1, \Lambda _{C_1})\cap {\rm Supp}(C_2, \Lambda _{C_2})=\emptyset,\]
which makes $\Q \bfc $ a \loc vector space.

\begin{prop}
The \loc set $\left(\Q \bfc, \top _D\right)$ equipped with the Minkowski sum is a \loc  algebra.
\end{prop}

\subsection { \Loc algebra of  decorated rooted forests.}
\label{subsec:forestsloc}

Let $(\Omega, \top _\Omega)$ be a locality set.
A {\bf properly $(\Omega,\top_\Omega)$-decorated  rooted forest} is a pair $(F, d)$, where $F$ is a (non-planar) rooted forest and $d: V(F)\to \Omega$ is a map from the set $V(F)$ of vertices of $F$ to $\Omega$ such that
$$v\not =v' \Rightarrow d(v)\top _\Omega d(v').
$$

Let $\calf _{\Omega,\top _\Omega}$ denote the set of properly $(\Omega, \top _\Omega)$-decorated rooted forests and by $K \calf _{\Omega,\top _\Omega}$ its linear span. The set $\calf _{\Omega,\top _\Omega}$ carries a natural locality relation $\top _{\calf _{\Omega,\top _\Omega}}$ from $(\Omega, \top _\Omega)$, and this locality relation induces a locality relation $\top _{\calf _{\Omega,\top _\Omega}}$ on  $K \calf _{\Omega,\top _\Omega}$.

\begin{rk}
 The symmetry of the concatenation of forests  motivates the symmetry of   the binary relation  on the decoration
 set $\Omega$. Planar forests call for the more general non symmetric \roc structure. Given a \roc-set $(\Omega,\top_\Omega)$, one could also define the algebra of {\bf $(\Omega,\top_\Omega)$-decorated planar
 forests} in a similar manner, yet taking care of preserving the order of the concatenation of vertices.
\end{rk}

\begin {prop} \cite[Proposition 1.22]{CGPZ2}
The free module  $K\, {\mathcal F}_{\Omega, \top_\Omega}$  of   properly $(\Omega, \top _\Omega)$-decorated rooted forests  is a  \loc algebra for the concatenation product.
\end{prop}

\subsection {Locality algebra of meromorphic germs of symbols} In analysis and geometry, the algebra of polyhomogeneous symbols plays an important role. We are in particular interested in the algebra $\cals (\R _{\ge 0})$ of polyhomogeneous symbols in $\R _{\ge 0}$.

For a  polyhomogeneous symbol $\sigma\in \cals (\R _{\ge 0})$ with asymptotic expansion
$$\sigma \sim \sum_{j=0}^{\infty }a_{ j}\, x^{\alpha-j},$$
the Hadamard finite part at infinity is defined by
\begin{equation}\label{eq:finitepart_const}
 \underset{+\infty}{\rm fp} \sigma:=\sum_{j=0}^\infty a_{j}\, \delta_{\alpha- j, 0},
\end{equation}
(with $\delta_{i,0}$ the Kronecker symbol). Unfortunately, $\underset{+\infty}{\rm fp}$ is not an algebra homomorphism on $\cals (\R _{\ge 0})$.

To make these structures compatible, we introduce the space $\Omega$ of meromorphic germs of symbols on $\R _{\ge 0}$ \cite {CGPZ2}, where germs are around $0$ in the filtered Euclidean space $ \left(\R ^\infty, Q\right)$, and define the locality relation $\perp ^Q$ similar to that of meromorphic functions induced by the inner product $Q$. Then we have

\begin{prop}\label{prop:orderprod} \cite[Proposition 4.15]{CGPZ2} The triple
		$\left(\Omega, \perp^Q, m_{\Omega}\right)$ is a  commutative and unital \loc algebra, with unit given by the constant function $1$ and
$m_{\Omega} $ is the restriction of	the pointwise function multiplication to  the graph $\perp^Q\subset \Omega\times \Omega$.
	\end{prop}

\section{\Loc morphisms}
\label{sec:locmor}
As indicated in the last section, we choose to work in the \loc setup and not the \roc-set framework to which many of the concepts below could be generalised.

\subsection {Basic notions and examples} Recall that 	
\begin{defn}
A {\bf \loc map} from a \loc set $\left(X,\mtop_X\right)$ to a \loc set $ (Y, \mtop_Y)$ is a map $\phi:X\to Y$ such that $(\phi\times \phi)(\mtop_X)\subseteq \mtop_Y$. More generally, maps $\phi,\psi:\left( X,\mtop_X\right)\to \left(Y, \mtop_Y\right)$ are called {\bf independent} and denoted $\phi\top \psi$ if
$(\phi\times \psi)(\mtop_X) \subseteq \mtop_Y$.
\mlabel{defn:localmap}
\end{defn}

So a \loc map is a map independent of itself.

\begin{defn}  \mlabel{defn:locallmap}
Let $\left(U,\mtop_U\right)$ and $ (V, \mtop_V)$ be \loc vector spaces. A linear map $\phi:\left( U,\mtop_U\right)\to \left(V, \mtop_V\right)$ is called a {\bf \loc linear map} if it is a \loc map.
\end{defn}

\begin {defn}
A \loc linear map $f:(A,\mtop_A,\cdot_A)\to (B,\mtop_B,\cdot_B)$ between two (not necessarily unital) \loc algebras is called a {\bf \loc algebra homomorphism} if
\begin{equation}
f(u\cdot_A v)=f(u)\cdot_B f(v)\ \  \tforall (u,v)\in\mtop_A.
\mlabel{eq:lmultlin}
\end{equation}
\end{defn}

By the definition, the composition of \loc morphisms is again a \loc morphism, so we have the category ${\bf LA}$ of \loc  algebras over $K$.

Here are fundamental examples of   \loc  morphisms on $\calm  $. The first one plays a central role in our multivariate minimal subtraction renormalisation scheme.

Since $\calm ^Q_{-}$ is a \loc ideal of $\calm $, we have
\begin{prop}\label{prop:locproj}  \cite[Proposition 3.19]{CGPZ1}{\bf (The $Q$-orthogonal projection onto holomorphic germs).} The   projection $\pi^Q_+:\left(\calm, \perp^Q\right) \longrightarrow \left(\calm _{+}, \perp^Q\right)$ is a \loc algebra homomorphism.
\end{prop}

Since $\top^Q _D\subset \perp ^Q$ and ${\rm Supp}(\pi ^Q_+f)\subset {\rm Supp}(f)$, the  projection $\pi^{Q}_+$ is also a \loc algebra homomorphism on   $\left(\calm , \top^Q _D\right) $.

\begin{rk} We view the fact of going from the locality relation $\perp^Q$ to the \loc relation  $\top_D^Q$  with a smaller graph $\top_D^Q\subset \perp^Q$,   as a reduction of the \loc relation, which rigidifies the setup in a manner similar to that fact that the structure group of a principal bundle to a subgroup rigidifies the underlying geometric setup.\end{rk}

\subsection{Locality morphisms on the algebra of  meromorphic  germs of symbols}
\label {subs:LocSym}
On the locality algebra $(\Omega, \perp^Q) $ of meromorphic germs of symbols, we can define several important locality maps:
\begin {itemize}
\item the Hadamard finite part at infinity map $\underset{+\infty}{\rm fp}: \Omega \to \calm$ ;
\item locality maps: ${\mathfrak S}_\lambda : \Omega \to \Omega$ with  $\lambda =0, \pm 1$.
\end{itemize}

These maps are constructed as follows. Though the Hadamard finite part at infinity map is not an algebra homomorphism on $\cals (\R _{\ge 0})$; yet   its extension to $\Omega$ enjoys the following property. 

\begin {prop} \label{prop:fpOmega}\cite[Proposition 4.17]{CGPZ2} The Hadamard finite part at infinity map $\underset{+\infty}{\rm fp}$ extends to a locality algebra homomorphism
 $$\underset{+\infty}{\rm fp}: (\Omega \perp ^Q)\to (\calm, \perp^Q). $$
\end{prop}

For $\lambda=1$ (resp. $\lambda=-1$) we define
$${\mathfrak S}_1(\sigma) (n):=\sum_{k=1}^n \sigma(k) \quad
 \text{\big(resp.} \ {\mathfrak S}_{-1}(\sigma) (n):=\sum_{k=1}^{n-1} \sigma(k)\big),$$
both maps can be interpolated by means of the Euler-MacLaurin formula \cite[Eqn. (13.1.1)]{H} to take values in $\Omega$; and  for $\lambda=0$ we define
 $${\mathfrak S}_{0}(\sigma) (x):={\mathfrak I}(\sigma) (x):=\int_{1}^x \sigma(y)\, dy.$$

\subsection {Locality morphisms on lattice cones}
\label {subs:LocCones}
On a strongly convex lattice cone $(C, \Lambda _C)$ with interior $C^o$,   discrete (resp. continuous) Laplace transforms of characteristic functions lead to exponential sums  (resp. integrals) and give rise to meromorphic functions
\[\sum_{\vec n\in C^o\cap \Lambda _C} e^{\langle \vec \e, \vec n\rangle} \quad  \quad \left({\rm resp.}\quad \int_{  C } e^{\langle \vec \e, \vec x\rangle}\, d\vec x _{\Lambda _C}\right).\]
These can be extended by linearity and subdivisions to any  convex lattice cone, to build
maps $S^o$ and $I$ from $\Q {\mathcal C}$ to $\calm$.

The idempotency $(C, \Lambda _C)\cdot (C, \Lambda _C)=(C, \Lambda_C)$ for any lattice cone $(C,\Lambda_C)$ implies that $S^o$ and $I$ are not algebra homomorphisms for the Minkowski sum $\cdot$, since otherwise they can only assume values $t$ with $t^2=t$, meaning $t=0 \ {\rm or }\ 1$. But in the locality setting, we have

\begin {prop} $($\cite[Proposition 3.7]{GPZ1}$)$ $S^o$ and $I$ are
 \loc algebra homomorphisms from $\left(\Q {\mathcal C}, \perp^Q\right)$ to $\left({\mathcal M}, \perp^Q\right)$.
\label{prop:SI}
\end{prop}

Similarly,  $S^o$ and $I$ are
 \loc morphisms from $\big(\Q {\mathcal C}, \top_D^Q \big)$ to $\left({\mathcal M}(\C^\infty), \top _D\right)$, a useful property  of these maps which shows the importance of locality algebra.

\subsection{Linear operators lifted to the algebra of rooted forests.}
\label {subs:LocLin}
Let us briefly recall some definitions and results borrowed from \cite{CGPZ3}.

\begin{defn}
Let $(\Omega, \top)$ be a \loc set. A {\bf \loc $(\Omega,\top)$-operated set} or simply a {\bf \loc operated set} is a \loc set $(X,\top_X)$ together with a {\bf partial action} $\beta$ of $\Omega$ on $X$: there is a subset
$\top_{\Omega,X}:=\Omega\times_\top X\subseteq \Omega \times X$
and a map

$$ \beta:  \Omega\times_\top X \longrightarrow X, \ (\omega, x)\mapsto \beta ^\omega (x) $$
satisfying the following compatibility conditions
\begin{enumerate}
\item
For
$$\Omega \times_\top X  \times_\top X: = \{(\omega,u,u')\in \Omega\times X\times X\,|\,
 (u,u')\in \top_X, (\omega, u), (\omega , u')\in \Omega \times_\top X\},$$
we have
$$\beta \times \Id_X: \Omega \times_\top X  \times_\top X
\longrightarrow X\times_\top X.$$
In other words,
If $(\omega,u,,u')$ is in $\Omega \times_\top X  \times_\top X$, then  $(\beta^\omega(u),u')$ is in $\top_X$.
  \item
For
$$\Omega\times_\top \Omega \times_\top X:=\{(\omega,\omega',u)\in \Omega \times \Omega \times X\,|\, (\omega,\omega')\in T_\Omega, (\omega,u), (\omega',u)\in \Omega\times_\top X\},$$
we have
$$ \Id_\Omega \times \beta: \Omega\times_\top \Omega \times_\top X \longrightarrow \Omega \times_\top X,$$
that is, if $(\omega, \omega')\in \top_\Omega, (\omega,u), (\omega',u)\in \Omega\times_\top X$, then
$(\omega',\beta^\omega(u))\in \Omega\times_\top X$.
	\end{enumerate}
 \mlabel{defn:locopset}
 \end{defn}

\begin{defn}
Let $(\Omega , \top )$ be a \loc set.
 \begin{enumerate}
  \item A {\bf \loc $(\Omega,\top)$-operated semigroup} is a  quadruple  $\left(U,\top _U, \beta ,m_U\right)$, where $(U,\top _U,m_U)$ is a \loc semigroup and
  $\left(U,\top _U, \beta \right)$ is a $(\Omega,\top)$-operated \loc set such that if
 $ (\omega, u, u')$ is in $\Omega\times_\top U\times_\top U$, then $(\omega, uu')$ is in $\Omega\times_\top U$;
 \item A {\bf \loc $(\Omega,\top)$-operated monoid} is a  quintuple  $\left(U,\top _U, \beta ,m_U,1_U\right)$, where $(U,\top _U,m_U,1_U)$ is an \loc monoid
  and $\left(U,\top _U, \beta, m_U \right)$ is a $(\Omega,\top)$-operated \loc semigroup such that $\Omega \times 1_U\subset \Omega \times _\top U$.
    \item A {\bf $(\Omega,\top)$-operated \loc nonunitary algebra } (resp. {\bf $(\Omega,\top)$-operated \loc unitary algebra}) is a  quadruple  $\left(U,\top _U, \beta ,m_U\right)$ (resp. quintuple
    $(U,\top _U, \beta,$ $m_U, 1_U)$) which is at the same time a \loc algebra (resp. unitary algebra) and a \loc $(\Omega,\top)$-operated semigroup (resp. monoid), satisfying the additional condition that for any $\omega\in \Omega$, the set $\{\omega\}^{\top_{\Omega,U}}:=\{ u\in U\,|\, \omega\top_{\Omega,U} u \}$ is a subspace of $U$ on which the action of $\omega$ is linear. More precisely, the last condition means
\begin{quote}
let $u_1, u_2\in U$. If $u_1, u_2\in \{\omega\}^{\top_{\Omega,U}}$ then for all $k_1, k_2\in K$, we have $k_1u_1+k_2u_2\in \{\omega\}^{\top_{\Omega,U}}$ and $\beta^\omega(k_1u_1+k_2u_2)=k_1\beta^\omega(u_1)+k_2\beta^\omega(u_2)$.
\end{quote}
(resp. this condition and $\Omega \times 1_U\subset \Omega \times _\top U$)
\end{enumerate}
\mlabel{defn:basedlocsg}
\end{defn}

\begin{defn}
Given $(\Omega,\top _\Omega)$-operated \loc structures  (sets, semigroups, monoids, nonunitary algebras, algebras) $(U_i, \top _{U_i}, \beta_i)$, $i=1,2,$ a {\bf morphism of \loc operated  \loc structures}  is a locality map $f: U_1\to U_2$ such that $f\circ \beta_1^\omega = \beta_2^{\omega} \circ f$ for all $\omega\in \Omega$.
\mlabel{defn:morphismoperatedloc}
\end{defn}

The key property of $K\, {\mathcal F}_{\Omega, \top_\Omega}$ is the following universal property.
\begin {prop} $K\, {\mathcal F}_{\Omega, \top_\Omega}$ is an $(\Omega, \top _\Omega)$-operated algebra, and it is the initial object in the category of commutative $(\Omega, \top _\Omega)$-operated algebra.
\label{pp:universal}
\end{prop}

Let $(\Omega, \top _\Omega)$ be a locality algebra. By the universal property of the initial object,
a linear map $\phi:\Omega \longrightarrow \Omega$ such that $\phi\top Id_\Omega$ induces a $(\Omega,\top_\Omega)$ \loc operation on itself, and $\phi$ lifts uniquely to a \loc morphism of $(\Omega,\top_\Omega)$-operated \loc algebra for this action \cite[{Corollary 1.24}]{CGPZ2}
\begin{equation}
{\widehat \phi}:  K {\mathcal F}_{\Omega, \top_\Omega} \longrightarrow \Omega.
\label{eq:hatphi}
\end{equation}

This can be applied  to   the space $(\Omega, \perp ^Q)$  of  multivariate meromorphic germs of symbols on $\R_{\geq 0}$. The  interpolated  summation map ${\mathfrak S}_\lambda$ on $\Omega$   (with  $\lambda=\pm 1$) give rise to branched maps
\begin{equation}\label{eq:locmortrees}\widehat{\frakS _\lambda} : (\R \calf _{\Omega, \perp ^Q}, \top _{\calf _{\Omega, \perp ^Q}})\to (\Omega , \perp ^Q)
\end{equation}
 \begin{prop} The  branched map $\widehat{\frakS _\lambda}$ is a locality  algebra homomorphism.
\end{prop}

\subsection {Operations lifted to the algebra of rooted forests}
\label {subs:operation}
An operation $\beta: \Omega \times U \longrightarrow U$ of a \loc set $(\Omega, \top_\Omega)$ on a \loc monoid $(U, \top_U)$ induces a \loc algebra morphism \cite[Proposition 2.6]{CGPZ3} \[\Phi_\beta: (\R {\mathcal F}_{\Omega, \top_\Omega}, \top_{{\mathcal F}_{\Omega, \top_\Omega}})  \longrightarrow  (U, \top_U).\]

On the filtered Euclidean space $\left(\R ^\infty, Q\right)$
we have a direct system
$$j_{n+1}^*: (\R ^n )^*\to (\R ^{n+1})^*.
$$
Let
$$\call =\varinjlim _n (\bbR ^n)^*$$
be the direct limit of spaces of linear forms.
To an  element $L\in \call$, regarded as a linear function on $\bbR ^\infty \otimes \bbC$, one   assigns a homogeneous pseudodifferential symbol   $x\longmapsto x^L$ on $(0,+\infty)$ of order $L$ i.e. for any $z\in \bbR ^\infty \otimes \bbC$, it defines a smooth function on $(0,+\infty)$ which is homogeneous in $x$ of degree $L(z)$. Let  \[\Omega:=\call; \quad U:=\calm [\call]\] where  $\calm [\call]$ is the group ring over $\calm$ generated by the  {additive} monoid $\call$, equipped with the  \loc relation $\perp^Q$ induced by that on $\calm$ and on ${\mathcal L}$:
$$\left(\sum_i f_ix^{L_i}\perp^Q \sum_j g_jx^{\ell _j}\right) \ \Longleftrightarrow \ \left( \{f_i, L_i\}_i\perp^Q \{g_j, \ell _j\}_j\right).
$$The map  $$ {\mathcal I}: (L,  f) \longmapsto \left(x\longmapsto \int_0^\infty  \frac {f(y)\,y^{-L}}{y+x} \,dy\right),$$ defines an operation  \begin{eqnarray}\label{eq:cali} {\mathcal I}:\call\times \calm [\call]&\longrightarrow& \calm [\call]\nonumber\\
(L, f)&\longmapsto&  {\mathcal I}(L, f)
\end{eqnarray}
  and can therefore be lifted to a   map \cite[Eqs.~(33)-(35)]{CGPZ3}\[\calr : (\R {\mathcal F}_{\call , \perp ^Q}, \top _{{\mathcal F}_{\call , \perp ^Q}})  \longrightarrow  (\calm [\call], \perp ^Q).\]
Composing the resulting map $\calr $   with the evaluation of the maps  at $x=1$, gives rise to a $ \calm$-valued \loc morphism
$$\calr _1=ev_{x=1}\circ \calr: (\R {\mathcal F}_{\call , \perp ^Q}, \top _{{\mathcal F}_{\call , \perp ^Q}})  \to (\calm, \perp ^Q).$$

\section{Renormalisation by locality morphisms}
\label{sec:renorm}
In this section we describe a general renormalisation scheme  via multivariate regularisations, and implement it to renormalise formal sums on lattice cones, branched formal sums and branched formal integrals.

\subsection {A renormalisation scheme}\label{subsec:renscheme}
If a regularised theory is realised by a \loc algebra morphism
\begin{equation}\label{eq:Phi} \phi^{\rm reg}: \left({\mathcal A}, \top\right)\longrightarrow \left(\calm, \perp^Q\right),\end{equation}
then by Proposition \mref {prop:m-}, the projection
$$\pi _+^Q: \calm \to \calm _+
$$
is a locality homomorphism. Therefore we have a \loc algebra homomorphism
\begin{eqnarray*}{\rm ev}_0 \circ \pi_+^Q \circ \phi^{\rm reg}:
 \left({\mathcal A}, \top\right)&\longrightarrow & \C,
 \end{eqnarray*}
where ${\rm ev}_0$ is the evaluation at $0$ and the locality relation on $\C $ is $\C \times \C$. This locality algebra homomorphism ${\rm ev}_0 \circ \pi_+^Q \circ \phi^{\rm reg}$ is taken as a renormalisation of $\phi^{\rm reg}$. This gives a
renormalisation scheme in this setting. When $(\cala,\top)$ is equipped with a suitable \loc Hopf algebra structure. This renormalisation agrees with the one from the algebraic Birkhoff factorisation~\cite[{Theorem 5.9}]{CGPZ1}.

\subsection {Renormalised conical zeta values}
\label{ss:renczv}

For a lattice cone $(C, \Lambda )$ in $(\R ^\infty, \Z ^\infty)$, were they well-defined, the formal sums
$$\sum _{n\in C^o\cap \Lambda _C}1 \quad  {\rm and}\  \sum _{n\in C\cap \Lambda _C}1 ,
$$
would yield values characteristic of the lattice cone, but they are unfortunately divergent. In order to extract information from these divergent expressions,  a univariate regularisation is shown to be less appropriate (see \cite {GPZ3}) than  multivariate regularisations, which appear as very natural:
$$S^o(C, \Lambda _C):=\sum _{n\in C^o\cap \Lambda _C}e^{\langle n,z\rangle} \ \text{ and }\  S^c(C, \Lambda _C):=\sum _{n\in C \cap \Lambda _C}e^{\langle n,z\rangle}.
$$
By subdivision techniques, we can extend $S^o$ and $S^c$ to linear maps from $\Q \bfc$ to $\calm$, which are locality algebra homomorphisms as discussed in Section \ref {subs:LocCones}. These are regularised maps for the formal expressions.

Therefore we have renormalised open conical zeta values for a lattice cone $(C, \Lambda _C)$
$$\zeta ^o(C, \Lambda _C):=({\rm ev}_0 \circ \pi_+^Q \circ S^o)(C, \Lambda _C)
$$
and renormalised closed conical zeta values for a lattice cone $(C, \Lambda _C)$
$$\zeta ^c(C, \Lambda _C):=({\rm ev}_0 \circ \pi_+^Q \circ S^c)(C, \Lambda _C).
$$
In fact, the function $ (\pi_+^Q \circ S^o)(C, \Lambda _C)$ or $(\pi_+^Q \circ S^c)(C, \Lambda _C)$ contains important geometry information for lattice cones -- they are building blocks of Euler-MacLaurin formula for lattice cones. Because of their geometric nature,  these formal expression can easily be renormalised by means of locality morphisms.

\subsection {Branched zeta values} \label{subsec:branchedzeta} In \cite{CGPZ2}, this multivariate renormalisation scheme was applied to
renormalise a branched generalisation of multiple zeta values. Renormalised multiple zeta values are related to the renormalisation of the formal sum
$$\sum _{n_1>\cdots > n_k{>0}}1.
$$
This formal sum is an  iterated sum corresponding to a   {totally ordered structure}, and   can therefore be viewed as a sum over  {ladder} trees. It generalises to  branched sums on more general partially ordered structures such as rooted trees.

In order to renormalise such  branched formal sums, we construct the regularisation maps from the locality morphisms in Section \ref {subs:LocLin} and Section \ref {subs:LocSym}
\begin{equation}\label{eq:Zlambda} {\mathcal Z}^\lambda=\underset{+\infty}{\rm fp}\circ \widehat{\frakS _\lambda} : (\R \calf _{\Omega, \perp ^Q}, \top _{\calf _{\Omega, \perp ^Q}}) \to (\calm, \perp ^Q).\end{equation}

Once    the regularisation is chosen,   a specific choice of meromorphic germs of symbols $x\longmapsto \sigma(s)(x):=\chi(x)\,x^{-s}$ on $\R_{\geq 0}$ \cite[Definition 5.1]{CGPZ2}, where $\chi$ is an excision function around zero,   leads to a generalisation of multiple zeta functions,
  namely {\bf regularised branched zeta functions}
  $$\zeta^{\rm reg, \lambda}: \R \calf _{\Omega, \perp ^Q} \to \calm .$$

Due to the locality of the morphisms involved in its construction, $\zeta^{\rm reg, \lambda}$ is a \loc morphism of \loc algebras. Composing on the left with the renormalised evaluation at zero ${\rm ev}_0\circ \pi_+^Q$ leads to  {\bf renormalised branched zeta values}
$$\zeta^{\rm ren, \lambda}: \R \calf _{\Omega, \perp ^Q} \to \R.$$

\subsection {Kreimer's toy model}\label{subsec:Kreimer} In \cite{CGPZ3}, this multivariate renormalisation scheme was applied to
 Kreimer's toy model \cite{K} which recursively assigns formal iterated integrals to rooted forests
induced by the formal grafting operator:
$$\beta _+(f)(x)=\int _0^\infty \frac {f(y)}{y+x}dy.
$$

There are different ways to regularise these divergent integrals. We adapt the regularisation by universal property of rooted forests studied in Section \ref {subs:operation}:
$$\calr _1: \R \calf _{\call, \perp^Q}\to \calm.
$$
Applying the renormalisation scheme to this locality map, we have the renormalised value to a properly decorated forest $(F,d)$
$$({\rm ev}_0 \circ \pi_+^Q \circ \calr _1 )(F,d).
$$

 We refer the reader to~\cite{CGPZ3} for further details.

\smallskip

\noindent
{\bf Acknowledgements.} The authors acknowledge supports from the Natural Science Foundation of China (Grant Nos. 11521061 and 11771190) and  the German Research Foundation (DFG project FOR 2402). They are grateful to the hospitalities of Sichuan University and University of Potsdam where parts of the work were completed.

\end{document}